\documentclass[10pt, oneside]{article}

\usepackage[top=2.5cm, bottom=2cm,left=2.2cm,right=2.2cm]{geometry}
\usepackage{amsmath, amsthm, amssymb}
\usepackage{lmodern}
\usepackage[T1]{fontenc}
\usepackage{fancyhdr}
\usepackage[all, cmtip]{xy}
\usepackage[titletoc,title]{appendix}
\usepackage[backend=biber,
            isbn=false,
            doi=false,
            maxbibnames=5,
            giveninits=true,
            style=alphabetic,
            citestyle=alphabetic]{biblatex}
\renewbibmacro{in:}{}
\let\OLDthebibliography\thebibliography
\renewcommand\thebibliography[1]{
  \OLDthebibliography{#1}
  \setlength{\parskip}{0pt}
  \setlength{\itemsep}{4pt plus 0.3ex}
}
\usepackage[usenames,dvipsnames]{color}
\usepackage{slashed} %For Dirac operators and gamma matrix nonsense
\usepackage{float}
\usepackage{caption}
\usepackage{subcaption}
\usepackage[pdftex, bookmarks=true, linkbordercolor={0 0 1}]{hyperref}
\usepackage{comment}
\usepackage{tikz, tikz-cd}
\usetikzlibrary{matrix}
\usetikzlibrary{shapes}
\usetikzlibrary{arrows}
\usetikzlibrary{calc,3d}
\usetikzlibrary{decorations,decorations.pathmorphing}
\usetikzlibrary{through}
\tikzset{ext/.style={circle, draw,inner sep=1pt, minimum size=5pt},
         int/.style={circle,draw,fill,inner sep=1pt, minimum size=5pt},nil/.style={inner sep=1pt}}
\tikzset{exte/.style={circle, draw,inner sep=3pt},inte/.style={circle,draw,fill,inner sep=3pt}}
\tikzset{diagram/.style={matrix of math nodes, row sep=3em, column sep=2.5em, text height=1.5ex, text depth=0.25ex}}
\tikzset{diagram2/.style={matrix of math nodes, row sep=0.5em, column sep=0.5em, text height=1.5ex, text depth=0.25ex}}

\theoremstyle{definition} \newtheorem{definition}{Definition}[section]

\newtheorem{lemma}[definition]{Lemma}
\newtheorem{theorem}[definition]{Theorem}
\newtheorem{prop}[definition]{Proposition}

\theoremstyle{definition} \newtheorem{remark}[definition]{Remark}
\theoremstyle{definition} 
\theoremstyle{definition} 
\theoremstyle{definition} \newtheorem{example}[definition]{Example}
\theoremstyle{definition}

\newcommand{\ZZ}{\mathbb{Z}}

\newcommand{\RR}{\mathbb{R}}
\newcommand{\CC}{\mathbb{C}}
\newcommand{\PP}{\mathbb{P}}

\newcommand{\TT}{\mathbb{T}}

\renewcommand{\gg}{\mathfrak{g}}
\newcommand{\hh}{\mathfrak{h}}

\newcommand{\su}{\mathfrak{su}}

\newcommand{\SU}{\mathrm{SU}}
\newcommand{\spin}{\mathrm{Spin}}

\def\d{{\rm d}}

\newcommand{\sub}{\subseteq}

\newcommand{\bb}[1]{\mathbb{#1}}

\newcommand{\mr}[1]{\mathrm{#1}}
\newcommand{\mc}[1]{\mathcal{#1}}
\newcommand{\mf}[1]{\mathfrak{#1}}

\DeclareMathOperator{\im}{im}

\DeclareMathOperator{\id}{id}

\DeclareMathOperator{\dvol}{dvol}

\fancypagestyle{bib}{
  \fancyhf{}
  \fancyhead[R]{\textsf{\nouppercase{\leftmark}}}
  \fancyhead[L]{\textsf{\textbf{\thepage}}}
}

\def\Sol{\mc{S}\rm{ol}}
\def\xto{\xrightarrow}

\title{Spontaneous symmetry breaking:\\ a view from derived geometry}
\author{Chris Elliott \and Owen Gwilliam}
\date{\today}

\begin{document}

\maketitle

\begin{abstract}
We examine symmetry breaking in field theory within the framework of derived geometry, as applied to field theory via the Batalin-Vilkovisky formalism.  Our emphasis is on the standard examples of Ginzburg--Landau and Yang--Mills--Higgs theories and is primarily interpretive.  The rich, sophisticated language of derived geometry captures the physical story elegantly, allowing for sharp formulations of slogans (e.g., for the Higgs mechanism, that the unstable ghosts feed the Goldstone bosons to a hungry, massless gauge boson).  Rewriting these results in the BV formalism provides, as one nice payoff, a reformulation of 't Hooft's family of gauge-fixing conditions for spontaneously broken gauge theory that behaves well in the $\xi \to \infty$ limit.

\end{abstract}

Symmetry breaking is a central concept in field theory,
especially gauge symmetry breaking as developed by Anderson \cite{Anderson}, Higgs \cite{Higgs1,Higgs2}, Englert, Brout \cite{EnglertBrout}, Guralnik, Hagen and Kibble \cite{GuralnikHagenKibble},
which plays a key role in explaining superconductivity and the massive gauge bosons of the strong and weak forces.
Our goal here is to revisit this notion, and these examples, using the framework of derived geometry.
In essence we translate the manipulations used in physics textbooks into the geometric language of stacks and tangent complexes,
and we find that this articulation clarifies some aspects.
We also expect it will help mathematicians better understand the concept;
at the least, we feel that any notion so crucial to physics deserves to be examined from every possible angle.

\section{An overview of the problem and our perspective}

We start by sketching the basic situation mathematically;
in the following sections we will unpack the quintessential example of symmetry breaking in both physical and mathematical styles.

Recall the broad outlines of classical field theory.
We begin with a space $\mc{F}$ whose elements we call the {\em fields}; 
typically $\mc{F}$ consists of sections of a fiber bundle over a (super)manifold.
There is also a function $S$ on the space $\mc{F}$ that we call the {\em action functional},
and according to the principle of least action, 
the space of true interest is the critical locus of $S$.
More accurately, one applies variational calculus to $S$
to construct a space $\Sol$ of solutions to the Euler-Lagrange equations.

We pause to remark that the term ``space''  means for us a derived stack,
but the reader need not be conversant with this notion.
For physicists we mention that, in essence, our geometric language will encode the BV/BRST formalism for field theory,
and little confusion will arise if they interpret things in this way.
For mathematicians we wish to emphasize that derived stacks provide a technical notion adequate for field theory,
although its application in physics is early in its development.

\begin{remark}
For experts in derived geometry, we note that derived algebraic geometry is not appropriate for (most of) field theory.
Instead it is usually more relevant to use derived differential geometry.
The focus in this paper is on formal aspects of derived geometry, however, 
and the formal aspects are the same in either setting.
For this reason we are willing to be somewhat cavalier about global aspects,
leaving a more precise analysis for future work~\cite{DDG}.
\end{remark}

Now suppose that a group $\mc{G}$ acts on $\mc{F}$ making the action functional $S$ equivariant for the $\mc{G}$-action.

There are two distinct ways in which we can treat this action.
\begin{enumerate}
\item One can view the action as a {\em symmetry} of the space $\Sol$.
\item One can form the quotient spaces $\mc{F}/\mc{G}$ and $\Sol/\mc{G}$ as derived stacks, that is (informally speaking) we view fields that are related by the action of $\mc G$ as ``indistinguishable''.  We say here that $\mc{G}$ is a {\em gauge symmetry}.
\end{enumerate}

\begin{remark}
In the most central examples for physics, including the Yang--Mills--Higgs theories that we will consider in this paper, gauge symmetries are \emph{local}, in the sense that the group action is allowed to vary over spacetime.  For instance, $\mc G$ may be the group of smooth maps from the spacetime manifold $X$ into a compact Lie group $G$.  Such local symmetries can be obtained from theories with a global $G$-symmetry by coupling to a new field: a connection on a principal $G$-bundle.  The combined procedure, where one first couples to connections then forms the quotient by local symmetries, is known as ``gauging'' a global symmetry.

The distinction that we mean to draw out here, however, is not between global and local symmetries, but between theories with a $\mc G$-action on the fields, and theories where the space of fields is a quotient stack.  Although in many physically interesting examples of classical field theories local symmetries are also gauge symmetries, these two dichotomies are not intrinsically the same.
\end{remark}

The role of $\mc{G}$ is quite different in these two cases,
but there are parallel questions to ask.
In particular, given the group action, it is natural to ask about the stabilizer at a given point $x$ in $\Sol$.
(For the gauge symmetry case, this stabilizer is thought of as the group of ``internal symmetries'' (or automorphisms) of the point $[x]$ in $\Sol/\mc{G}$.)
As one might expect, the behavior of the physical state of the system encoded by a point varies significantly depending on its stabilizer.
The key to understanding symmetry breaking is to analyze the orbits and their stabilizers:
in a sense that we make precise, symmetry breaking is about how the physics changes as the stabilizers change.

Indeed, central ideas like Goldstone bosons and the Higgs mechanism will arise just by studying the linearization of this question about orbits and stabilizers,
but linearization is more sophisticated in the setting of derived geometry. 
To whit, a tangent space to a derived stack -- i.e., its linearization at some point -- is encoded by a cochain complex known as the {\em tangent complex}.
(Physicists will have seen such complexes when describing the free theories underlying the BV/BRST descriptions of gauge theories.)
In the setting of a broken symmetry, this enhancement matches in a simple way with the usual discussion of Goldstone bosons.
We analyze this case primarily to develop the reader's understanding of how the derived formalism captures the standard approach.

In the setting of a broken gauge symmetry, however, the derived formalism provides a nice mathematical mechanism. The key point is that cochain complexes are meaningful in derived geometry only up to quasi-isomorphism,
and this feature lets one search for another theory, perhaps with a simpler description, whose tangent complex is quasi-isomorphic to that of the initial theory.
For instance, when a Higgs boson ``breaks'' a Yang--Mills theory with group $G$ to another Yang--Mills theory with subgroup~$H$, we interpret this to mean that there is a quasi-isomorphism of linearizations between the $G$-theory and the $H$-theory.

\begin{remark}
This perspective on gauge symmetry breaking suggests a moral: 
for a full understanding of gauge theory, 
it is important to remember the data of the moduli \emph{stack} $\Sol/\mc{G}$ and not only the coarse moduli space.  
To phrase the moral in BV/BRST language, 
we benefit from considering the full space of BRST fields, including ghosts, and without immediately choosing a gauge fix.  
Indeed, when we break the gauge symmetry from $G$ to $H$, the presence of the ghosts (i.e. the generators of infinitesimal symmetries in the tangent space to the stack) will be essential to construct a quasi-isomorphism to the $H$-theory.
\end{remark}

We now put this perspective to work in a few classic examples.
We will assume the reader has a passing familiarity with either the BV/BRST formalism or the language of stacks.  In fact, these two formalisms are related, as we will very briefly explain in Section~\ref{sec: BV review}: the BV/BRST formalism can be thought of as encoding the formal geometry near a point in a derived stack.  For an introduction to the language of derived stacks, we refer the reader to To\"en's survey articles \cite{ToenOverview,ToenSurvey}.  The interpretation of the BV/BRST formalism as encoding formal derived geometry is discussed in~\cite[Chapter 5]{CG2}.

\begin{remark}
The perspective we will take on the Higgs mechanism is through perturbation theory, which should be distinguished from the axiomatic notion of the Higgs phase of a gauge theory, in terms of the asymptotic behaviour of the expectation values of Wilson loop operators \cite{tHooftPhase} (see also the discussion in~\cite[Lecture 7]{WittenIAS}).
\end{remark}

\subsection{The BV formalism and formal derived geometry}
\label{sec: BV review}

Perturbation theory around a fixed solution --- in the derived framework mentioned above ---
recovers the classical field theory in the BV formalism.
For an extensive discussion and motivation of this idea, which is not tautological, see \cite{CG2}.\footnote{David Carchedi and the second author \cite{DDG} are providing a global setting in derived differential geometry that formalizes completely this idea.}
Here we give a very brief review, 
as we build on this perspective.

In terms of global geometry, we pick a point in~$\Sol/\mc{G}$ and wish to describe the deformation theory (or formal neighborhood) of this chosen solution.
In physical language, this amounts to understanding the perturbation theory of that solution.
In derived geometry, such deformation theory is always encoded by a shifted $L_\infty$ algebra,\footnote{An $L_\infty$ algebra is to Lie algebras as an $A_\infty$ algebra is to associative algebras. These provide a flexible, powerful tool that generalizes Lie theory from vector spaces to cochain complexes. A {\em shifted} $L_\infty$ algebra is a cochain complex that is an $L_\infty$ algebra when shifted up by one degree. These appear naturally in physics \cite{Stasheff, BFLS}.}
a folklore result now established in a global setting \cite{Hinich, LurieSAG, Pridham}.
Remarkably, when one works out the shifted $L_\infty$ algebra associated to perturbations around a solution for a classical field theory,
one finds a graded vector space consisting of extended fields, namely the fields and their antifields, as well as --- if needed --- ghosts, antighosts, and so on.
The brackets of the shifted $L_\infty$ algebra arise from the Taylor expansion of the action functional around the chosen solution.
This perfect dictionary between BV perturbation theory and derived deformation theory means that a physicist can think in terms of the extended action functional (including all the BV/BRST fields) while a mathematician can think of a shifted $L_\infty$ algebra, without any loss of insight.

As a quick example, consider a pure Yang--Mills theory on an $n$-dimensional oriented manifold $X$ for a compact Lie group $G$ and a choice of real-valued nondegenerate invariant pairing on its Lie algebra~$\gg$.
Consider perturbations of the solution given by the trivial flat connection on the trivial principal $G$-bundle over $X$,
which leads to the standard BV treatment of pure Yang--Mills theory. 
The usual graded vector space of fields is concentrated in cohomological degrees $-1,0,1$ and 2:
\[
\Omega^0(X, \gg)[1] \oplus \Omega^1(X, \gg) \oplus \Omega^{n-1}(X, \gg)[-1] \oplus \Omega^n(X, \gg)[-2].
\]
If one shifts the degrees up by one, it is easy to see a dg Lie algebra structure:
$\gg$-valued 0-forms act on the other forms by wedging the form components and bracketing the $\gg$-components,
and there is a differential that 
\begin{itemize}
\item applies the exterior derivative $\d$ to a 0-form and an $n-1$-form but
\item applies $\d \star \d$ to a 1-form.
\end{itemize}
All this data can be read from the quadratic part of the Yang--Mills action functional along with a knowledge of the gauge symmetries.
(Note that it is {\em not} contained in just the action functional.)  We can include additional brackets that further encode the cubic and quartic parts of the Yang--Mills action functional

The tangent complex $\TT_0$ at this trivial connection is precisely the underlying cochain complex of the dg Lie algebra.  The shift $\TT_0[-1]$ of this complex is equipped with an $L_\infty$ structure that includes the action of $\gg$-valued 0-forms on the other forms, in addition to degree 2 and 3 brackets encoding the part of the action functional of Yang--Mills theory of degree 3 and 4 respectively.  We discuss this explicitly in Appendix~\ref{bracket_appendix}.

\begin{remark}
The antifields and antighosts show up when studying variational calculus from a derived point of view: they generate spaces of relations between the ordinary fields, which impose the equations of motion. (See \cite{CG2} for a motivational discussion; \cite{DDG} will contain a systematic, global treatment.) 
\end{remark}

Now, instead let us consider perturbations of a solution given by a connection $\d + A$ on the trivial bundle, where now $A$ may be non-zero.
What happens here is that the differential changes:
\begin{itemize}
\item it is $\d_A = \d+A$ on a 0-form and an $n-1$-form and
\item it $\d_A \star \d_A$ on a 1-form,
\end{itemize}
but the Lie bracket is unchanged.

The tangent complex $\TT_A$ at this connection is precisely the underlying cochain complex of the deformed dg Lie algebra.
Note that the cohomology of this complex depends on~$A$.
One calls the connection $\d + A$ {\it irreducible} if $\mathrm H^{-1} = 0$, 
because it means the stabilizer is trivial (or at least discrete),
but reducible if there is nontrivial $\mathrm H^{-1}$.

\begin{remark}
Note that this condition is slightly weaker than the usual notion of irreducibility, which says that a connection $A$ on a principal $G$-bundle is irreducible if the holonomy group of $A$ is all of $G$.  The cohomology group $\mathrm H^{-1}$ is the Lie algebra of the centralizer of the holonomy group, so $\mathrm H^{-1}$ is trivial when the centralizer is discrete, which is weaker than being irreducible when $G$ has discrete center (see e.g.~\cite[pp 132--133]{DonaldsonKronheimer}.) 
\end{remark}

For the theories considered in this paper -- gauge theories, possibly coupled to fermions or a scalar theory -- 
an important feature holds:
the underlying graded vector space of fields is unchanged, no matter the choice of solution,
but the shifted $L_\infty$ brackets vary depending on the solution.
Hence the tangent complex at any solution consists of the fixed graded vector space of fields but with a differential depending on the solution.

\begin{remark}
Our approach uses the formal neighborhood of a solution and hence can only see infinitesimal information.
It cannot, for instance, distinguish between a gauge theory with group $SU(2)$ and a theory with group $SO(3)$.
As we will see below, many features of gauge symmetry breaking can be seen even when this group-level information is washed out,
but it can be interesting to remember a little more.
Indeed, it is believed that remembering the stabilizer {\em group}, and not just its Lie algebra, explains the appearance of  topological gauge theories with {\em discrete} gauge groups as phases of nontopological gauge theories (for a starting point, see \cite{GukKap}).
It is possible to refine our approach via derived geometry to accommodate this feature,
which is joint work in progress with David Carchedi.
In brief, given some solution $\phi_0 : \ast \to \Sol/\mc{G} =: \mc{X}$ to the equations of motion viewed as a point on a derived stack, 
the map $\phi_0$ factors canonically through its formal neighborhood $\mc{X}^{\wedge}_{\phi_0} \hookrightarrow \mc{X}$.
If $H = {\rm Stab}(\phi_0) \subset \mc{G}$, then we also have a map of derived stacks $\text{pt}/H \to \mc{X}$ through which $\phi_0$ factors.
The formal neighborhood of {\em this} map knows about both perturbation around $\phi_0$ and the relevant {\em group} of gauge symmetries~$H$.
\end{remark}

\subsection{An incomplete discussion of the literature}

The theory of symmetry breaking and the Higgs mechanism is a foundational topic in classical and quantum field theory,
appearing in every quantum field theory textbook. 
A particularly complete account appears in \cite{Weinberg2}, 
and we have referred to \cite{PesSch, Ryder} in preparing this paper.  
Mathematicians might appreciate a summary~\cite{Hamilton} aimed at them.  

Due to its importance, there is much work by mathematical physicists,
often rooted in algebraic or constructive quantum field theory, 
which has a far more analytic flavor than what we do here.
See~\cite{StrBook}  for a systematic treatment by Strocchi with many references.

Our approach, following the physics textbooks, is perturbative in nature.  
This approach is via a systematic application of the BV formalism, which derives from foundational work of Batalin and Vilkovisky \cite{BatalinVilkovisky} in the physics literature.
(We recommend Witten's succinct and illuminating treatment \cite{WittenAntibracket} as a nice starting point.)
Our perspective on the BV formalism is heavily informed by the work of Costello \cite{Costello},
who stressed its meaning and content from the perspective of derived geometry.
This view was pioneered by Stasheff and collaborators in the 1990s \cite{Stasheff, BFLS},
who emphasised the role of $L_\infty$ algebras.  
Recently, this reformulation of field theories using the language of $L_\infty$ algebras, 
particularly in the string field theory literature, is receiving active attention~\cite{HohmZwiebach, AHHL, CG2}.

In just this kind of setting, \cite{Zeitlin} gives an extensive examination of abelian Yang--Mills and Yang--Mills--Higgs theories, 
with a focus on relations to string field theory.
We strongly encourage the reader to read this paper as well, because it touches on a number of interesting issues.
In particular, Section 4.3 of {\it loc. cit} discusses how to see the abelian Higgs mechanism with a single complex scalar Higgs field,
writing the same complex as we do here and noting the acyclic subcomplex involving the ghosts.
Our paper extends and complements Zeitlin's treatment, 
particularly via the discussion of global derived geometry.

\section{Breaking a global symmetry}
\label{sec: global}

Let us consider a field theory on which a group $G$ acts by symmetries.
Since the space of all solutions is preserved by the action of $G$,
each solution $\phi_0$ lives on some $G$-orbit, which has the form $G/H$ with $H = {\rm Stab}(\phi_0)$.
Note that the variational derivative $\delta S$ of the action must vanish in the orbit directions, but not in the directions transverse to the orbit.

When we examine the formal neighborhood of $\phi_0$ in the space of all solutions,
it is natural to ask how the Lie algebra $\mf h = {\rm Lie}(H)$ acts on this formal neighborhood and whether it organizes and simplifies analysis of the theory.
At the level of the linearized problem, in the BV formalism,
we could ask how to decompose the tangent complex $\TT_{\phi_0}$ into irreducible $\hh$-representations.
Tangent directions to the orbit -- known as Goldstone modes -- must be $\hh$-invariant and are typically massless,
whereas the transverse directions transform in non-trivial representations of $\hh$, and will typically have interesting linearized equations of motion.

We will demonstrate these ideas using the example of Ginzburg--Landau theory, in part because this will be a useful example upon which to build in the next section, where we examine gauge symmetry breaking.  We will also discuss briefly how to think about the theory of twists of supersymmetric theories through the lens of global supersymmetry breaking.

\subsection{The Ginzburg--Landau example} \label{GL_section}

We will work on a smooth spacetime manifold $X$, equipped with a Riemannian metric.  Let $G$ be a compact Lie group, and let $R$ be a finite-dimensional $G$-representation equipped with a $G$-invariant symmetric pairing $\langle -, - \rangle \colon R^{\otimes 2} \to \RR$.
Our fields consist of a scalar field $\phi \in C^\infty(X) \otimes R$ and a potential function $V \colon R \to \RR$ that is $G$-invariant.  
Our theory has action functional
\[
S(\phi) = \int_{X} \left(|\d\phi|^2 + V(\phi)\right) \dvol,
\]
and the equations of motion are 
\[
\Delta \phi + \nabla V|_{\phi} = 0
\]
where $\Delta$ denote the Laplacian on $X$ (tensored with the identity on $R$) and $\nabla V$ denotes the gradient of $V$ as a function on $R$, 
which we then evaluate at the point~$\phi(x) \in R$.

For simplicity, we will focus on a solution $\phi_0$ that is constant, and hence is just given by a point in $R$.
By the equations of motion, we see that $\phi_0$ must be a critical point of the potential $V$ on $R$.
For the sake of concreteness, and because it is conventional in physics textbooks,
from now on we fix our potential to be the quartic functional
\[
V(r) = - \frac 14(|r|^2 - m^2)^2,
\]
where $m$ is a fixed real number. 
The critical points then lie on the sphere of radius~$m^2$, together with the point at the origin.

\begin{remark}
This example generalizes easily to a nonlinear $\sigma$-model where the target is a $G$-manifold $Y$.
When expanding perturbatively around a classical solution $\phi_0 \colon X \to Y$, 
one studies a field theory where a field $\phi$ is a smooth section of the vector bundle $\phi_0^* TY \to X$,
and the potential function is defined fiberwise but may vary along the base~$X$.
We will not describe how the symmetry breaking extends to this more general setting, but the same techniques we demonstrate apply.
\end{remark}

We first describe the BV theory expanded around the solution $\phi_0 = 0$,
although our focus will ultimately be on perturbative theories arising from other solutions.
The graded vector space of fields is
\[
\begin{matrix}
\Omega^0(X; R) \oplus \Omega^n(X; R)[-1].
\end{matrix}
\]
The tangent complex $\TT_0$ is
\[
\Omega^0(X; R) \xto{\Delta-m^2} \Omega^n(X; R)
\]
in degrees 0 and 1, which arises from the purely quadratic terms in the action when expanded at this trivial solution.

Now let us expand the action functional around $(\phi_0)$ where $\phi_0$ is constant and $|\phi_0|^2 = m^2$.
We will use the variable $\varphi$ for $\phi - \phi_0$ so that
\begin{align*}
S'(\varphi) 
&= S(\varphi + \phi_0) \\
&= \int_{X} \left(|\d\varphi|^2 - \frac 14(|\varphi + \phi_0|^2 - m^2)^2 \right) \dvol \\
&= \int_{X} \left(|\d\varphi|^2 - \frac 14(|\varphi|^2 + 2 \langle \phi_0, \varphi \rangle)^2 \right) \dvol \\
&= \int_{X} \left(|\d\varphi|^2 - (\langle \phi_0, \varphi \rangle)^2 - \langle \phi_0, \varphi \rangle |\varphi|^2 - \frac 14|\varphi|^4 \right) \dvol
\end{align*}
The quadratic part of the action functional now contains a term $(\langle \phi_0, \varphi \rangle)^2$ with very interesting properties.
Namely, decompose $R$ into $R' \oplus R''$ where $R'$ consists of vectors orthogonal to $\phi_0$ and $R'' = \RR\langle \phi_0 \rangle$ consists of vectors ``parallel'' to $\phi_0$.  Then we can write $\varphi$ as a sum $\varphi' + \varphi''$,
and we find
\[
(\langle \phi_0, \varphi \rangle)^2 = (\langle \phi_0, \varphi'' \rangle)^2.
\]
In other words, the $\varphi'$ component is {\em massless} while the $\varphi''$ component is {\em massive}.

Note that $R'$ contains the linearized $G$-orbit of $\phi_0$, in other words, no element of $G$ acts by rescaling $\phi_0$.  Indeed, if $\phi_0$ is an eigenvector of $g$ with eigenvalue $\lambda$, then by invariance of the pairing, $\langle \phi_0, \phi_0 \rangle = \langle g \phi_0, g\phi_0\rangle = \langle \lambda \phi_0, \lambda \phi_0 \rangle$, so $\lambda = \pm 1$.  This means that infinitesimal variations of the scalar field along the $G$-orbit are massless.
That is, the Goldstone bosons are components of the massless field~$\varphi'$.

The tangent complex $\TT_{\phi_0}$ is now just
\[
\xymatrix{
\Omega^0(X; R') \ar[r]^{\Delta} &\Omega^n(X; R') \\
\Omega^0(X; R'') \ar[r]^{\Delta - |\phi_0|^2 } &\Omega^n(X; R'') \\
},
\]
which arises from the purely quadratic terms in the action functional~$S'$.
The leftmost column is in degree zero, and the rightmost column is in degree one.

\subsection{Twisting a supersymmetric theory as global symmetry breaking} 
\label{susy_section}

The notion of \emph{twisting} a supersymmetric field theory, as introduced by Witten \cite{WittenTQFT}, can be seen as an instance of global supersymmetry breaking, at least in the form we articulate in this paper.  
In this section we will briefly give an outline of this idea in the BV formalism, following the approach taken in~\cite{ElliottSafronov, ESW}.

\begin{remark}
The procedure we will describe in this section is different from supersymmetry breaking as the term is usually used.  Rather than passing to perturbative field theory around a classical solution that is not supersymmetry-invariant,  
we will consider instead perturbation theory around a non-zero value of an element of the supersymmetry group, 
viewed as an auxiliary (or background) field.
\end{remark}

\begin{definition}
Fix a spinorial representation $\Sigma$ of $\spin(n)$, and a non-zero $\spin(n)$-equivariant linear map $\Gamma \colon \Sigma^{\otimes 2} \to \RR^n$, where $\RR^n$ carries the defining representation.  The \emph{supertranslation group} associated to the pair $(\Sigma, \Gamma)$ is the super Lie group whose underlying supermanifold takes the form $\RR^n \times \Pi \Sigma$, where $\Pi$ indicates the shift into odd degree, and where the group operation is given by the map $\Gamma$.
\end{definition}

\begin{definition}
A classical field theory on $\RR^n$ is \emph{supersymmetric} if it is equipped with a global action of a supertranslation group extending the action of $\RR^n$ by translations.
\end{definition}

Such theories have fermionic fields, so in the BV formalism we now have a bigrading $\ZZ \times \ZZ/2\ZZ$,
where the $\ZZ$ provides a cohomological (or ghost) grading and the $\ZZ/2\ZZ$ provides a boson/fermion grading. 

Consider a supersymmetric theory whose BV formulation is given by a shifted $L_\infty$ algebra $\bb T$, 
which encodes the perturbation theory around some fixed classical solution.  Informally, the \emph{twist} of $\bb T$ by an element $Q$ of $\Sigma$ is a theory where the action functional is modified so that it encodes how $Q$ acts on the theory.
Concretely, we deform the differential and $L_\infty$ brackets on~$\bb T$.  

Alternatively, we can rephrase twisting as a case of symmetry breaking.  
We extend the field content of the BV theory by including an ``extra field'' $Q$, 
namely as a constant function on spacetime valued in the vector space $\Sigma$
and placed in $\ZZ \times \ZZ/2\ZZ$-degree $(-1,1)$. 
We extend the action functional $S$ to a functional $S_Q$ that depends on this background field in $\Sigma$,
so that the $\Sigma$-action is encoded by this action functional.  
(Concretely, we have $\{S_Q, -\} = \{S, -\} + Q$.)  The result is that we can view this extended theory as a family of theories parametrized by $Q \in \Sigma$, 
or, in more sophisticated language, as a sheaf $\bb T_{\mc N}$ of $L_\infty$ algebras over the locus $\mc N$ inside the vector space $\Sigma$ consisting of elements $Q \in \Sigma$ such that $\Gamma(Q,Q)=0$.  
This setup lets us efficiently define twisting.

\begin{definition}
Let $Q \in \Pi \Sigma$ be an odd element of the supertranslation group such that $\Gamma(Q,Q) = 0$.  
The \emph{twist} of the theory $\bb T$ by $Q$ is the perturbative field theory $\bb T_Q$ around $Q$, i.e. the fiber of the coherent sheaf $\bb T_{\mc N}$ at the point $Q$ in $\mc N$.
\end{definition}

Note that it was essential that $Q$ was an odd symmetry, so that the vector space $\Sigma$ in our family was placed in even total degree.

In the paper \cite{ESW}, using this language, the first author, Safronov and Williams studied all the possible twists of supersymmetric Yang--Mills theories.  
For each such theory, they construct a perturbative equivalence of classical BV theories between the twisted theory and a more understandable field theory, typically a version of BF theory.  
These calculations are very similar to the calculations carried out in this paper, 
where an equivalence of cochain complexes is constructed by ``integrating out'' certain acyclic summands in the cochain complex associated to the twisted field theory.

\begin{remark}
By the procedure discussed above, we can think of a supersymmetric field theory as defining a sheaf (in fact, a constructible sheaf) 
of $L_\infty$ algebras over the variety $\mc N$ consisting of all square-zero elements in~$\Sigma$: 
this space $\mc N$ is called the \emph{nilpotence variety} in~\cite{EagerSaberiWalcher}.  
The twist by $Q$ is the fiber of this sheaf over the element $Q$. 
It is in this sense that we often refer to a twist as a ``sector'' of a supersymmetric field theory: 
it includes those quantities in the full family over $\mc N$ that are visible in the specified fiber.
\end{remark}

\section{Breaking a gauge symmetry}
\label{sec: gauge}

We now turn to the more interesting situation: breaking of gauge symmetries.
We will spend most of our time analyzing two classic examples, but before delving into them, we outline the basic setup.

In brief, consider a theory with gauge symmetries, which has a global derived stack $\Sol/\mc{G}$ of solutions up to gauge equivalence.
If one examines the BV theory at a solution $x$,
its linearization around $x$ corresponds to a free theory.
Let $\TT_x$ denote the cochain complex encoding this free theory.
In examples of gauge symmetry breaking, 
what one sees is that $\TT_x$ changes qualitatively as $x$ moves;
in fact, the cohomology $\mr H^\bullet(\TT_x)$ can jump.
As we vary $x$, some components of the space of fields may become Nakanishi--Lautrup fields --- to use physical language ---  
which means that fields are paired off with corresponding components in the space of ghosts.
In mathematical language, as one varies $x$, the complex $\TT_x$ may develop acyclic subcomplexes, 
which may be removed by passing to a quasi-isomorphic but smaller complex.

From the perspective of the global derived stack, one sees how the stabilizer of $x$ in $\mc G$ changes as $x$ varies.
After linearization, ghosts orthogonal to the stabilizer subgroup make some Goldstone bosons exact in the cochain complex.
Thus, only the ghosts of the stabilizer subgroup and a subspace of the bosons survive at the level of cohomology.
This perspective provides a mechanism for how the massless gauge boson becomes massive by eating the Goldstone boson: {\it the unstable ghosts {\em feed} Goldstone bosons to that hungry, massless gauge boson}.

We will analyze this mechanism in an example of primary interest to physics:
in a Yang--Mills theory coupled to a boson with an interesting potential,
the perturbative theory around some points is equivalent (that is, the tangent complex is quasi-isomorphic) to a theory in which the gauge bosons are massive. 
Below we revisit this classic example in the BV formalism.
At the linearized level, the statement is quite simple: 
there is an inclusion $T \hookrightarrow \TT_x$ of cochain complexes (in fact, a quasi-isomorphism!)
where $T$ has connection 1-forms in degree 0, 
but the differential acting on them has a mass term, i.e. the differential operator has a zeroth order coefficient.
Conceptually, one is applying a change of coordinates to the derived space of fields
such that the cohomologically relevant information is encoded in the subspace $T$ of the fields.
This quasi-isomorphism extends to an equivalence of the full BV theories as perturbative classical theories.
In short, {\em the Higgs mechanism is an example of an equivalence of BV theories}.

We will demonstrate these ideas using the Yang--Mills--Higgs model,
which is the classic setting for gauge symmetry breaking.
We will briefly recall this theory as it is usually specified before describing it in the BV formalism 
and analyzing symmetry breaking with an abelian and a non-abelian gauge group.

\subsection{The Yang--Mills--Higgs model}

We will study here a Yang--Mills--Higgs theory with compact gauge group $G$ and representation $R$ equipped with a $G$-invariant non-degenerate inner product $\langle -,- \rangle \colon R^{\otimes 2} \to \RR$. 
We will work on the trivial principal $G$-bundle over a smooth oriented $n$-dimensional manifold $X$ equipped with a Riemannian metric.
It is easy to extend our results to nontrivial bundles by simply using the appropriate associated bundles everywhere.  For comparison, a textbook account of the usual physics approach to symmetry breaking in Yang--Mills--Higgs theory is presented in \cite[\S 21]{Weinberg2}.

Our fields consist of a $G$-gauge field $A \in \Omega^1(X,\gg)$, an $R$-valued scalar field $\phi \in C^\infty(X) \otimes R$, and a potential function $V \colon R \to \RR$ that is $G$-invariant.  
Our theory has action functional \footnote{More generally we could include a gauge coupling constant by defining the covariant derivative $\d_A$ to be $\d + gA$, where $g$ is a real constant.}
\[
S(A, \phi) = \int_{X} \left(-\frac 14|F_A|^2 + \frac 12 |\d_A(\phi)|^2 + \frac 12 V(\phi)\right) \dvol.
\]
The equations of motion say that 
\begin{align*}
\d_A \star F_A &= \langle \phi, \d_A \phi \rangle \\
\star \d_A \star \d_A \phi + \nabla V|_{\phi} &= 0.
\end{align*}
The reader should bear in mind that there is a group $\mc{G} = {\rm Maps}(X,G)$ of gauge automorphisms that acts on the space of solutions to the equations of motion.

In fact, we should take the {\em derived} space of solutions,
but we will not elaborate on that aspect here.
Instead, we will take it explicitly into account when we describe the formal derived neighborhood of a solution 
via the BV formalism.

We now describe the BV theory expanded around the solution $(A_0 = 0, \phi_0 = 0)$,
although our focus will ultimately be on perturbative theories arising from other solutions.
The graded vector space of fields is
\[
\begin{matrix}
\Omega^0(X; \gg)[1] \oplus \Omega^1(X;\gg)  \oplus \Omega^{n-1}(X; \gg)[-1]  \oplus\Omega^n(X; \gg) [-2] \\
\oplus \\
\Omega^0(X; R) \oplus \Omega^n(X; R)[-1].
\end{matrix}
\]
The first row was already described in Section~\ref{sec: BV review}, and it has the same shifted $L_\infty$ structure.
The second row -- by itself -- was already described in Section~\ref{sec: global},
but here it is a module over the first row by extending the $\gg$-action on the representation $R$ in a natural way 
by wedging the form factors; this is precisely the action of the gauge transformations and connections on the associated bundles.  There are further quadratic and cubic $L_\infty$ brackets between the first and second rows coming from the cubic and quartic terms in the action functional. 

The tangent complex $\TT_{\phi_0}$ is concentrated in cohomological degrees $-1$ to $2$, and is given by
\[
\xymatrix@R-=2pt{
\bb T_{\phi_0,-1} \ar[r]^{\d_{-1}} \ar@{=}[dd] &\bb T_{\phi_0,0} \ar[r]^{\d_{0}} \ar@{=}[dd] &\bb T_{\phi_0,1} \ar[r]^{\d_{1}} \ar@{=}[dd] &\bb T_{\phi_0,2} \ar@{=}[dd] \\ &&& \\
\Omega^0(X; \gg) \ar[r]^{\d} &\Omega^1(X;\gg) \ar[r]^{\d \star \d } &\Omega^{n-1}(X; \gg) \ar[r]^{\d} &\Omega^n(X; \gg) \\
&\oplus &\oplus & \\
&\Omega^0(X; R) \ar[r]^{\d \star \d- m^2 \star} &\Omega^n(X; R) &
},
\]
which arises from the purely quadratic terms in the action when expanded at this trivial solution.
Note that the kernel of the differential on the second row corresponds to the solutions of the PDE $\Delta \phi - m^2 \phi = 0$.

For the rest of this section, we will expand around solutions of the form $(A_0 = 0, \phi_0)$ where $\phi_0$ is a constant function and hence is determined by an element in $R$.
Looking at the equations of motion, we see that for $\phi_0$ to be a solution, it must be a critical point of the potential $V$ on $R$.
For the sake of concreteness, and because it is conventional in physics textbooks,
from now on we fix our potential to be the quartic functional
\[
V(r) = -\frac 14(|r|^2 - m^2)^2,
\]
where $m$ is a fixed real number, and $|r|^2$ is defined using the pairing on the representation~$R$. 
The critical points then lie on the sphere of radius~$m^2$.

If one expands the action functional around $(0,\phi_0)$ where $\phi_0$ is constant and $|\phi_0|^2 = m^2$,
then the mass terms appear in the quadratic part of the action functional.
In particular, consider the quadratic part that does not involve any derivatives;
it now takes the form
\[
\int_{X} \left(\frac 12 |A\wedge  \phi_0|^2 -\frac 12 m^2|\phi|^2  \right) \dvol,
\]
so that both $A$ and $\phi$ look massive.

We now turn to a more thorough examination of this example.

\subsection{Abelian Yang--Mills theory with a Higgs boson}

First we will consider the case where $G = \mr U(1)$, in which many terms vanish.  
Furthermore take $R$ to be a 1-dimensional complex representation of weight~$k$.  
Our first step is to describe the formal derived geometry of the derived stack $\Sol/\mc{G}$ at a solution $(0, \phi_0)$ where $\phi_0$ is a point in $R$ with radius~$m^2 = |\phi_0|^2$.

The tangent complex $\TT_{\phi_0}$ is
\[
\xymatrix{
\Omega^0(X) \ar[r]^{\d} \ar[dr]_{i\phi_0} &\Omega^1(X) \ar[r]^{D_1} \ar[dr]^(.35){D_2} &\Omega^{n-1}(X) \ar[r]^{\d} &\Omega^n(X) \\
&\Omega^0(X; R) \ar[ur]_(.35){D_3} \ar[r]_{D_4} &\Omega^n(X; R) \ar[ur]_{2k\langle i\phi_0, - \rangle} &
}
\]
where the operators $D_i$ arise from the quadratic part of the action functional when expanded around this solution.  Note that a factor of $i$ occurs in the infinitesimal gauge algebra action because we have identified forms valued in the Lie algebra of $\mr U(1)$ as real, rather than imaginary, differential forms.  We find
\begin{align*}
D_1 &= -\d \star \d + 2k^2m^2 \star \\
D_2 &=  - ki \phi_0 \,\d \star \\
D_3 &= -\star 2ki\langle \phi_0 \, , \d -\rangle\\
D_4 &= \d\star \d - \pi_{\phi_0}m^2\star,
\end{align*}
where $\pi_{\phi_0}$ is the orthogonal projection onto the real span of $\phi_0$ in $R$.  We immediately note the similarity to the Ginzburg--Landau example of Section \ref{GL_section}: the massive scalar field $\phi$ splits into a piece parallel to $\phi_0$, which is still massive (the Higgs boson) and a piece perpendicular to $\phi_0$ which is now massless (the Goldstone boson).

Look at how the cohomology of this complex compares to the cohomology at the point $\phi_0 = 0$ (which is an isolated singularity of the potential $V$).
For instance,  $\mr H^{-1}(\TT_0) = \RR$ but $\mr H^{-1}(\TT_{\phi_0}) = 0$ for $\phi_0 \ne 0$.  
Concretely, one sees that the stabilizer group changes.

In fact, a much stronger statement can be made:
there is another classical BV theory that is {\em equivalent} to this theory -- as perturbative theories --
but is simpler in nature.
The underlying na\"ive fields are a connection 1-form $B \in \Omega^1(X)$ and a real scalar $f \in C^\infty(X)$,
and the action functional is
\[
S'(B,f) =  \int_X -\frac 14 |\d B|^2 + M^2 |B|^2 - m^2f^2 + \frac 12 |\d f|^2 + \text{higher order}.
\]
The mass of the field $B$ here is given by $M = km$. 
Ignoring the terms involving $f$, this massive boson $B$ is governed by the {\em Proca action functional},
which is commonly introduced by physics textbooks in precisely this setting.
We emphasize that there is {\em no gauge symmetry} acting on $B$ here, 
so we are interested in just the derived variational locus of $S'$ and not some quotient thereof.
This $B$ field is coupled to a scalar field $f$ that takes values in $\RR$ rather than~$\RR^2$, by contrast with~$\phi$.
The coupling is encoded in higher order terms that we did not spell out.

\begin{remark}
In physics this theory is called a massive gauge theory but the name is something of a misnomer because there is no gauge symmetry.
The $B$ field, however, is a connection 1-form, leading to the name.
\end{remark}

The tangent complex $\TT'$ of this theory at the solution $(B_0=0, f_0=0)$ is
\[
\xymatrix@R-20pt{
\Omega^1(X) \ar[r]^{D_1} & \Omega^{n-1}(X) \\
\Omega^0(X) \ar[r]^{D_4} &\Omega^n(X)
},
\]
placed in degrees 0 and 1.
Note that there are no ghosts or antighosts.

The tangent complexes $\TT_{\phi_0}$ and $\TT'$ are quasi-isomorphic. 
Even better, one can retract $\TT_{\phi_0}$ onto~$\TT'$.

\begin{prop} 
\label{abelian_retract_prop}
The complex $\TT'$ is a deformation retract of $\TT_{\phi_0}$.
\end{prop} 

The proof of this proposition is below, but let us note an important consequence of having such an explicit way to transfer information between these two descriptions: there is an equivalence of formal derived spaces, and hence perturbatively, between the two BV theories at hand, in the following sense.

\begin{definition}[{\cite[Definition 1.20]{ESW}}]
Two classical BV theories $\bb T$ and $\bb T'$ are \emph{perturbatively equivalent} if there is a map $F \colon \bb T[-1] \to \bb T'[-1]$ of $L_\infty$ algebras whose underlying map of cochain complexes is a quasi-isomorphism intertwining the shifted symplectic pairings.
\end{definition}

Note that by a ``map of $L_\infty$ algebras,'' we mean here the homotopy-coherent notion (i.e., corresponding to a strict map of their dg cocommutative coalgebras).

\begin{theorem}
The classical BV theory encoded by $S'$ at the point $(0,0)$ is perturbatively equivalent to the classical BV theory encoded by $S$ at the point~$(0,\phi_0)$.
\end{theorem}

This result follows from Proposition \ref{abelian_retract_prop} because we can use standard transfer formulae \cite[Chapter 10]{LodayVallette} for shifted $L_\infty$ algebras to equip the complex $\TT'$ with a quasi-isomorphic shifted $L_\infty$ structure.
In physical terms, this means we can equip the massive gauge and scalar fields with an action functional 
whose dynamics are equivalent to those of the original theory, 
when expanded around the nonzero field~$\phi_0$.
(Note that the shifted symplectic structures are intertwined by the maps as well:
see Lemma~\ref{intertwine_brackets_lemma}.)

We now turn to proving the proposition.
In physical terminology, we ``integrate out'' the doublet $(c, \chi)$, 
consisting of the ghost $c$ and the Goldstone boson $\chi$ (the component of $\phi$ with the same phase as $i\phi_0$).
In mathematical terms, these form a contractible subcomplex 
as the differential in the classical BV complex included an isomorphism $c \to \chi$ given by scaling by~$ki\phi_0$.
With a bit of thought, this process can be described quite generally in the BV formalism
(cf. Proposition 1.23 of~\cite{ESW}).

\begin{proof}[Proof of Proposition \ref{abelian_retract_prop}]
We will define cochain maps $i,p$ and $h$ with the following sources and targets:
\[
\xymatrix{
\bb T' \ar@<3pt>[r]^\iota & \bb T_{\phi_0} \ar@<3pt>[l]^p \ar@(dr,ur)[]_h
}.
\]
Let $\alpha = \phi_0/|\phi_0|$ denote the phase of $\phi_0$.  The map $p$ is defined by
\begin{align*}
p( A,  \phi) &= ( A - \langle\phi_0^{-1}, \d \phi\rangle, \mr{Im}(\alpha^{-1} \phi)) \text{ in degree } 0 \\
p( A^*,  \phi^*) &= ( A^*, \mr{Im}(\phi_0^{-1} \phi^*) - (\d A^*)\phi_0) \text{ in degree } 1
\end{align*}
and the map $\iota$ is defined by
\begin{align*}
\iota( A,  \phi) &= ( A + \langle \phi_0^{-1},\d  i\alpha \phi\rangle, i\alpha \phi) \text{ in degree } 0 \\
\iota( A^*,  \phi^*) &= ( A^*, i\phi_0 \phi^* + i\phi_0^2\d A^*) \text{ in degree } 1.
\end{align*}
The composite $p \circ \iota$ is equal to the identity on $\TT'$.  The difference $\iota \circ p - \id_{\TT_{\phi_0}}$ is given by the map sending $( A,  \phi)$ to $(-\mr{Re}(\phi_0^{-1}\d( \phi)), -\phi_0\mr{Re}(\phi_0^{-1} \phi))$, and sending $( A^*,  \phi^*)$ to $(0, - \mr{Re}(\phi_0^{-1}-\phi_0  \d( A^*)) - \mr{Re}(\phi_0^{-1}  \phi^*))$. 

The cochain homotopy $h$ has degree $-1$. 
From degree 0 to degree -1, it is given by 
\[
h( \phi) = \mr{Re}(\phi_0^{-1}  \phi)
\] 
on $\Omega^0(X)$, 
and from degree 2 to degree 1, it is given by
\[
h(c^*) = \mr{Re}(\phi_0^{-1}c^*)
\] 
in $\Omega^n(X;R)$.  
This map satisfies the condition $[\d_{\TT_{\phi_0}}, h] = \iota \circ p - \id_{\TT_{\phi_0}}$, as required.
\end{proof}

These maps play nicely with the natural $(-1)$-shifted symplectic structures on $\TT_{\phi_0}$ and on $\TT'$ given by the wedge pairing of differential forms and the invariant pairing on~$R$.  
Direct computation verifies the following claim.

\begin{lemma} \label{intertwine_brackets_lemma}
The maps $\iota$ and $p$ intertwine these pairings up to rescaling the scalar term by $-\phi_0^{-2}$,
namely, 
\[
\langle \iota( A,  \phi), i( A^*,  \phi^*) \rangle_{\TT_{\phi_0}} = \langle  A,  A^* \rangle_{\TT'} - \phi_0^{-2}\langle  \phi,  \phi^* \rangle_{\TT'},
\] 
and a similar expression for~$p$.
\end{lemma}

\subsection{The General Case} \label{general_gauge_breaking_section}

Now let's do the general case, allowing for nonabelian gauge group $G$ and arbitrary finite-dimensional representation $R$. 
As before, fix a constant function $\phi_0$ from $X$ to $R$ where $\phi_0$ is a critical point of the potential~$V$. 
Let's describe the geometry at play here.
Inside $R$, we have the critical locus ${\rm Crit}(V)$, which is preserved by the $G$-action by hypothesis.
Thus we have a threefold inclusion
\[
R \supset {\rm Crit}(V) \supset G \cdot \phi_0 
\]
where $G \cdot \phi_0$ denotes the $G$-orbit of $\phi_0$.
For simplicity, let us assume that $\phi_0$ is a smooth point of ${\rm Crit}(V)$ and also a smooth point of the (possibly immersed) submanifold $G \cdot \phi_0$.
Then the linearized version of the inclusions above is
\[
T_{\phi_0} R \supset T_{\phi_0}{\rm Crit}(V)  \supset T_{\phi_0}(G \cdot \phi_0)
\]
Note that $T_{\phi_0} R \cong R$ as vector spaces but not as $\gg$-representations.
In physical terms, the subspace $T_{\phi_0}{\rm Crit}(V)$ encodes the massless directions along which the field $\phi$ can vary (to first order), so these are the Goldstone modes.
The subspace $T_{\phi_0}(G \cdot \phi_0)$ encodes the massless directions that are trivialized by (first-order) gauge transformations, so these are the Goldstone modes eaten by ghosts.
We now use these observations to construct another BV theory.

Let $H \subset G$ denote the stabilizer of $\phi_0$, and let $\hh \sub \gg$ be the infinitesimal stabilizer of $\phi_0$.
There is a Lie algebra action of $\hh$ on $T_{\phi_0} R$,
and it has an invariant subspace $T_{\phi_0}(G \cdot \phi_0) \cong \gg/\hh$.
The $G$-invariant pairing on $R$ induces a non-degenerate, $\hh$-invariant inner product on $R \cong T_{\phi_0} R$ (which is simply a vector space isomorphism), 
and we decompose $R$ into $R_{\phi_0} \oplus R_\perp$ as $\hh$-representations.
Here $R_{\phi_0}$ denotes $T_{\phi_0}(G \cdot \phi_0)$, the tangent space to the $G$-orbit of $\phi_0$ in $R$,
and $R_\perp$ denotes its orthogonal complement. %which are the elements $\phi$ in $T_{\phi_0} R$ such that $\langle \phi, \phi_0 \rangle = 0$.
As a vector space $R_\perp \cong T_{\phi_0} R/T_{\phi_0}(G \cdot \phi_0)$, so it encodes the directions that are not trivialized by gauge transformations.
Note that $\gg$ also decomposes as a direct sum $\hh \oplus \hh_\perp$ of orthogonal subspaces using a non-degenerate pairing invariant on~$\gg$.  

The underlying cochain complex $\TT_{\phi_0}$ of the perturbative theory near $\phi_0$ is
\[\xymatrix{
\Omega^0(X; \gg) \ar[r]^{\d} \ar[dr]_{\phi_0} &\Omega^1(X;\gg) \ar[r]^{D_1} \ar[dr]^(.35){D_2} &\Omega^{n-1}(X; \gg) \ar[r]^{\d} &\Omega^n(X; \gg) \\
&\Omega^0(X; R) \ar[ur]_(.35){D_3} \ar[r]_{D_4} &\Omega^n(X; R), \ar[ur]_{\phi_0} &
}\]
where in defining the arrows from the second row to the first we used the invariant pairings on $R$ and $\gg$ to derive an equivariant map $\langle -,-\rangle_\gg \colon R \otimes R \to \gg$ from the action map $R \otimes \gg \to R$.  The operators $D_i$ are given by formulae analogous to those we saw in the abelian case.  

Let us assume, for simplicity, that the $G$-orbit of $\phi_0$ coincides with $\mr{Crit}(V)$, so that $T_{\phi_0}(G \cdot \phi_0) = T_{\phi_0} \mr{Crit}(V)$.  Then we have:
\begin{align*}
D_1 &= -\d \star \d + 2\langle\phi_0,\phi_0\rangle \star \\
D_2 &= - \alpha_{\phi_0} \d \star \\
D_3 &= - \star 2\langle \phi_0, \d -\rangle_\gg \\
D_4 &= \d\star \d + \pi_{R_\perp}m^2 \star,
\end{align*}
where $\pi_{R_\perp}$ is the projection onto the $\hh$-subrepresentation $R_\perp$ of $R$, and $\alpha_{\phi_0}$ is the map from $\gg \to R$ sending $x$ to $x \cdot \phi_0$.  There are also quadratic and cubic Lie brackets given by the action of $\Omega^0(X; \gg)$ on everything,
along with similar brackets to those mentioned in the abelian case.  
There are also quadratic and cubic brackets coming from the Lie bracket on $\gg$ that did not appear above.

\begin{remark}
 In order to recover the example where $\gg = \mf u(1)$ and $R = \CC$ from the previous section, note that in that derivation we used a specific convention under which $\mf u(1)$ was identified with $\RR$ and the infinitesimal action of $A \in \mf u(1)$ on $\phi \in R$ is via $A \cdot \phi = iA\phi$.  In particular, the operators $D_2, D_3$ incuded a factor of $i$ and $\pi_{R_\perp}$ was given by projection onto the real subspace orthogonal to $i\phi_0$.
\end{remark}

We define a classical field theory that will model Yang--Mills--Higgs theory after symmetry breaking, with the following linearization:
\[
\TT' = 
\xymatrix@R-20pt{
\Omega^0(X; \hh) \ar[r]^{\d} &\Omega^1(X;\hh) \ar[r]^{\d \star \d}  &\Omega^{n-1}(X; \hh) \ar[r]^{\d} &\Omega^n(X; \hh) \\
&\Omega^1(X;\hh_\perp) \ar[r]^{D_1}  &\Omega^{n-1}(X; \hh_\perp) & \\
&\Omega^0(X; R_\perp) \ar[r]^{D_4} &\Omega^{n}(X; R_\perp). &
}
\]
Just as in the abelian case, we have the following.

\begin{prop} \label{def_retract_prop}
The complex $\TT'$ is a deformation retract of~$\TT_{\phi_0}$.
\end{prop} 

\begin{proof}
We will define cochain maps $i,p$ and $h$ with the following sources and targets.
\[
\xymatrix{
\TT' \ar@<3pt>[r]^i & \TT_{\phi_0} \ar@<3pt>[l]^p \ar@(dr,ur)[]_h
}
\]

We can define the maps $i$ and $p$ in the following way.  On the first row of the complex $\TT'$ the maps are defined by the inclusion $\hh \to \gg$ and the projection $\gg \to \hh$.  On the second and third rows they are defined in a similar way to Proposition \ref{abelian_retract_prop}.  So, let us write $( A,  \phi)$ for an element of $\Omega^1(X; \hh_\perp) \oplus \Omega^0(X; R_\perp)$, and $( A^*,  \phi^*)$ for the corresponding antifields in $\Omega^{n-1}(X; \hh_\perp) \oplus \Omega^n(X; R_\perp)$.  

Let $\pi_{R_\perp}$ denote the projection $R \to R_\perp$, let $\iota_{R_\perp}$ denote the inclusion $R_\perp \to R$, and let $\pi_{\hh_\perp}$ and $\iota_{\hh_\perp}$ denote similarly the inclusion of and projection onto $\hh_\perp \sub \gg$.  Acting on $\phi_0$ defines a map $\gg \to R$, and we can use the pairings to dualize this to a map $\alpha^\vee_{\phi_0} \colon R \to \gg$.  We then define
\begin{align*}
p( A,  \phi) &= ( A - \pi_{\hh_{\perp}}(\alpha^\vee_{\phi_0}(\d( \phi))), \pi_{R_\perp}( \phi)) \\
p( A^*,  \phi^*) &= ( A^*, \pi_{R_\perp}( \phi^*) - \d( A^*)\phi_0))\\
\iota( A,  \phi) &= ( A + \iota_{\hh_\perp}(\alpha^\vee_{\phi_0}(\d( \phi))), \iota_{R_\perp}( \phi))  \\
\iota( A^*,  \phi^*) &= ( A^*, \iota_{R_\perp}( \phi^*) + \d ( A^*)\phi_0).
\end{align*}
We can likewise define the homotopy $h \colon \TT_{\phi_0} \to \TT_{\phi_0}[-1]$ by $h( \phi) = \alpha^\vee_{\phi_0}(\pi_{\RR \phi_0} \phi)$ in $\Omega^0(X;\gg)$ (from degree 1 to degree 0), and $h(c^*) = c^* \phi_0$ in $\Omega^n(X;R)$, the map given by acting on $\phi_0$ (from degree 3 to degree 2).  One can then verify that this data satisfies the necessary conditions to define a deformation retract.
\end{proof}

As in the abelian case, we can check that the symplectic structures on $\TT_{\phi_0}$ and $\TT'$ are compatible up to ``rescaling by $\phi_0$''.  
Here this means that, for instance
\[
\langle \iota( A,  \phi), \iota( A^*,  \phi^*) \rangle_{\TT_{\phi_0}} = \langle  A,  A^* \rangle_{\TT'} - \langle \alpha^\vee_{\phi_0}( \phi), \alpha^\vee_{\phi_0}( \phi^*) \rangle_{\TT'},
\] 
and
\[
\langle p( A,  \phi), p( A^*,  \phi^*) \rangle_{\TT'} = \langle ( A)\phi_0, ( A^*)\phi_0 \rangle_{\TT_{\phi_0}} + \langle  \phi,  \phi^* \rangle_{\TT_{\phi_0}}.
\] 
That the map is a symplectomorphism is key to the main theorem.  Just as in the abelian case, we can use homotopy transfer along the deformation retract to extend the free perturbative field theory encoded by $\TT'$ to an interacting theory with action functional $S'$.  Altogether, we have shown the following.

\begin{theorem}
The classical BV theory encoded by $S'$ at the point $(0,0)$ is perturbatively equivalent to the classical BV theory encoded by $S$ at the point~$(0,\phi_0)$.
\end{theorem}

\begin{remark}
Mathematical readers might enjoy knowing that this theorem applies to a situation of genuine physical interest: the theory of the electroweak force.
This theory starts as a Yang--Mills--Higgs theory where the gauge group is $G = \mr U(1) \times \SU(2)$ and where the gauge field $\phi$ is valued in the complex two dimensional representation $R$ on which the $\SU(2)$ factor acts by its fundamental representation and the $\mr U(1)$ factor acts with weight 1.  When we work around a solution where $\phi$ takes a non-zero value $\phi_0$, the resulting theory is perturbatively equivalent to a theory with classical BV complex $\bb T'$ as in Proposition \ref{def_retract_prop}, where
\begin{itemize}
 \item The unbroken Lie algebra $\hh = \mr{Stab}(\phi_0)$ is a copy of $\mf u(1)$, embedded in $\gg$ according to the \emph{Gell-Mann--Nishijima formula}: that is, the embedding $\theta \mapsto \left(\theta, \frac \theta 2 h \right) \in \mr u(1) \oplus \su(2)$.  The corresponding massless gauge field is the photon.
 \item The remaining factor of the gauge Lie algebra $\hh_\perp$ is three dimensional.  The three corresponding massive gauge bosons are known as the $W^\pm$ and $Z^0$ bosons.
 \item The remaining factor $R_\perp$ of $R$ is real one dimensional.  The corresponding scalar field is known as the Higgs boson.
\end{itemize}
We have not included fermions in this discussion, but doing so is straightforward.
\end{remark}

\begin{remark}
We can construct a family of theories that interpolates between the two homotopy equivalent descriptions of the gauge theory with broken symmetry.  Consider the one parameter family $\TT_t$ of classical field theories on $X$ with the following linearization:
\[
\xymatrix{
\Omega^0(X; \hh) \ar[r]^{\d} &\Omega^1(X;\hh) \ar[r]^{\d \star \d}  &\Omega^{n-1}(X; \hh) \ar[r]^{\d} &\Omega^n(X; \hh) \\
\Omega^0(X; \hh_\perp) \ar[dr]_{\phi_0} \ar[r]^{t\d} &\Omega^1(X;\hh_\perp) \ar[r]^{\d \star \d + t|\phi_0|^2\star} \ar[dr]^(.35){t^2D_2}  &\Omega^{n-1}(X; \hh_\perp) \ar[r]^{t\d} & \Omega^n(X, \hh_\perp) \\
&\Omega^0(X; R_0) \ar[ur]_(.35){t^2D_3} \ar[r]_{t^3D_4} &\Omega^{n}(X; R_0) \ar[ur]_{\phi_0} \\
&\Omega^0(X; R_\perp) \ar[r]^{D_4} &\Omega^{n}(X; R_\perp).
}
\]
If we set $t=1$ then we obtain the complex $\TT_{\phi_0}$. If instead we set $t = 0$ we obtain the sum of the complex $\TT'$ with a contractible cochain complex. 
\end{remark}

\begin{example}
More elaborate examples of gauge symmetry breaking appear throughout physics.
As an example consider ``hidden local symmetry,'' 
which was introduced as an attempt to better understand effective chiral theories arising from QCD.
(See the exposition in \cite{MaHarada}.)  
Recall that in the chiral theory, the flavor group $G_F$ acts both from left and right so we can extend the global symmetry to $G_F \times G_F$.
This symmetry breaks to a residual diagonal symmetry $G_F$,
so that the broken theory becomes a sigma model with target $G_F \cong G_F \times G_F/G_F^{diag}$.
In practice people consider the perturbative expansion around a constant field $\phi = g \in G$ in the target.
The goal of hidden local symmetry is to show this sigma model is perturbatively equivalent to some theory with gauged flavor symmetry, which we now describe.

We now set $G = G_F$ for brevity's sake.  Consider the perturbative sigma model
with a $\gg \oplus \gg$-valued scalar field $(x,y)$ and action functional
\[
S(x,y) = \int \frac 12|\d x - \d y|^2 + V(x-y) \dvol,
\]
where $V$ is a potential.  
Here we will take $V$ to be the usual quartic (which might appear as a truncation of a more sophisticated model).
Note that $\gg$ acts diagonally on the target and hence provides a global symmetry of this theory.
 
Let us now enhance this theory by gauging this global symmetry.
Add a $\gg$-valued gauge field $A$.
We then obtain a variant of Yang--Mills--Higgs theory with representation $\gg \oplus \gg$: 
the theory with action functional
\[
S(A,x,y) = \int -\frac 14|F_A|^2 + \frac 12|\d_A x - \d_A y|^2 + V(x-y)  \dvol.  
\]
When we work perturbatively around a generic critical point of the action functional, 
this gauge symmetry breaks.
The gauged theory is perturbatively equivalent to a theory with massless $\gg$-valued scalar fields (playing the role of the pions)
and a massive $\gg$-valued 1-form field (playing the role of the vector mesons).
\end{example}

\begin{remark}[Twisting supergravity]
In parallel to the story of twisting supersymmetric field theories told in Section~\ref{susy_section}, 
there is a setting where the supersymmetry acts locally, i.e., where the action varies in spacetime.  
A notable class of theories with local supersymmetry are supergravity theories, 
i.e., supersymmetric extensions of theories of classical gravity.  
Costello and Li \cite{CostelloLi} recently introduced a theory of twisting for supergravity theories; 
as we will explain, their prescription is yet another example of gauge symmetry breaking in the language we have discussed. 

Let a shifted $L_\infty$ algebra $\bb T$ describe such a supergravity theory on $\RR^n$ in the classical BV formalism. Among the fields is a bosonic field $q \in C^\infty(\RR^n; \Pi \Sigma)[1]$ referred to as a ``bosonic ghost,'' 
where $\Pi \Sigma$ is a space of supertranslations, placed in odd $\ZZ/2\ZZ$-degree.  
Because $q$ has overall even degree, it makes sense to consider symmetry breaking, 
where we consider perturbation theory around a classical solution where $q$ takes a non-zero value.
For instance, we might have constant $q$ such that $q^2 = 0$ in the supertranslation group.  
Such perturbative theories define \emph{twists} of the supergravity theory.
\end{remark}

\section{Towards quantization of Yang--Mills--Higgs theories}

It is beyond the scope of this paper to construct the perturbative quantization of the BV theory for the spontaneously broken gauge symmetry.  For those familiar with the textbook story, however,
we rapidly sketch how the 't Hooft ({\it aka} $R_\xi$) gauges appear in our articulation of the theories.
With those in hand, it is possible to deploy the usual power-counting arguments to verify renormalizability {\it \`a la} 't Hooft-Veltman.

\subsection{Recollections on gauge-fixing}

We will describe the concept of gauge-fixing in concrete terms before turning to its realization in the BV formalism.  It will be enough to restrict attention to free theories in this section.

Let $F_0$ denote the space of ordinary (as opposed to BRST) fields, and let $F_{-1}$ denote the space of ghosts.  The subscript will end up representing cohomological degree in the BV formalism.
(In our examples above, $F_0$ is typically the space of $\gg$-valued 1-forms, encoding connection 1-forms,
and $F_{-1}$ is $\gg$-valued functions.)
We have a linear map $d \colon F_{-1} \to F_0$ that encodes the infinitesimal gauge symmetry.
Geometrically, we view the quotient space $F_0/\im(d)$ as the ``true'' space of interest, 
since it identifies fields that are gauge-equivalent.
We also have a linear map $D \colon F_0 \to F_1$, 
where $Df = 0$ is the linear equation of motion for this theory.
(Here $F_1$ is just some auxiliary vector space to phrase the equations of motion.)
We require that $D(dc) = 0$ for all $c \in F_{-1}$,
so that $\im(d) \subset \ker(D)$.
For the classical theory, the ``true'' solutions to the EoM are $\ker(D)/\im(d)$,
identifying naive solutions that are gauge-equivalent.
For the quantum theory, we want to integrate over the quotient space $F_0/\im(d)$, 
which is unfortunate as often we only know how to integrate over $F_0$ itself. 

In either the classical or quantum setting, we have {\em quotient} spaces but it would be nicer to work with subspaces.
Hence, it is convenient to pick a ``slice'' inside $F_0$, 
which is a linear subspace $V \subset F_0$ such that 
\begin{itemize}
\item $V$ and $\im(d)$ span $F_0$ and 
\item the intersection $\ker(D) \cap V$ is isomorphic to the true solutions $\ker(D)/\im(d)$ of the equation of motion.
\end{itemize}
We call such a subspace $V$ a choice of {\em gauge-fix}.

\begin{remark}
It is usually possible to pick out this linear subspace as the critical locus of a {\em quadratic} function.
This function is the {\em gauge-fixing Lagrangian} $L_{\mr{GF}}$,
and physicists add it to the na\"ive action $L$ on $F_0$.
Taking the critical locus of the sum $L + L_{\mr{GF}}$, 
one finds the intersection~$\ker(D) \cap V$
(This is not guaranteed in general but holds in standard examples.)
\end{remark}

In the BV formalism we extend the graded vector space to include antifields and antighosts.
That is, we have a cochain complex
\[
F_{-1} \xto{d} F_0 \xto{D} F_0^\vee \xto{d^\vee} F_{-1}^\vee
\]
where $F_0^\vee$ plays the role of $F_1$ from above.
``Formal solutions'' to the equation of motion are identified with the cohomology of this complex.
We will now use $D$ to denote the differential on the whole complex, rather than giving a distinct name for each map.

Again, we want to find a slice of $F_0$ that is complementary to $\im(d)$
in the sense that together they span~$F_0$.
Instead of specifying a subspace $V \subset F_0$,
we ask for a map $D^* \colon F_0 \to F_{-1}$
such that $\ker(D^*)$ is a gauge-fix in the sense already articulated.

But we want to pursue this approach compatibly with the full structure of the cochain complex.
Hence it is natural to ask for a degree -1 operator $D^*$ on the entire graded vector space of fields, ghosts, and their antitheses such that
\begin{itemize}
\item $(D^*)^2 = 0$, 
\item in each degree $k$, the intersection $\ker(D) \cap \ker(D^*)$ is isomorphic to the cohomology .
\end{itemize}
This kind of situation shows up in the Hodge decomposition of the de Rham complex on a Riemannian manifold,
and so a lot of intuition can be borrowed from those results: one can require that the commutator $[D, D^*]$ to be a generalized Laplacian and then employ Hodge theory \cite[Chapter 7.4]{Costello}.

\subsection{The 't Hooft gauges}

There is a well-known family of gauge-fixes in the setting of gauge symmetry breaking for Yang--Mills--Higgs theory:
for a real number $0 < \xi < \infty$, consider the gauge-fixing Lagrangian
\begin{align*}
L_{\mr{GF}}(A, \phi) 
&= \frac{1}{2\xi}\int (\partial_\mu A^\mu + \xi \langle \phi_0, \phi \rangle)^2\\
&= \frac{1}{2\xi}\int (\star \d \star A + \xi \langle \phi_0, \phi \rangle)^2
\end{align*}
whose critical locus picks out the subspace $\{ \star \d \star A + \xi \langle \phi_0, \phi \rangle = 0 \}$ inside the degree 0 fields.
An appealing feature of this family of gauges is that the propagators behave well and lead to power-counting proofs of renormalizability.

It is convenient to rephrase this family of subspaces so that we can take the $\xi \to \infty$ limit cleanly.
Set
\[
V_{(s:t)} = \{ s \star \d \star A + t \langle \phi_0, \phi \rangle = 0 \}
\]
where $(s:t)$ denotes a point on the projective line~$\RR\PP^1$. 
The point $[1:\xi]$ corresponds to the gauge-fixing introduced above.
Note that $V_{[0:1]}$ corresponds to the condition $\langle \phi_0, \phi \rangle = 0$.
In other words, one works only with the scalar fields that live in $R_\perp$,
which are the physically relevant fields as we have seen.
This gauge-fixing is known as the {\em unitarity gauge}.

We now give a BV extension of this gauge.

\begin{lemma}
There is a gauge-fixing operator $D^*_\xi$ for the spontaneously broken BV theory, for $0 < \xi < \infty$, extending the usual 't Hooft gauges.  We define $D^*_\xi$ to be
\[
\xymatrix{
\Omega^0(X; \hh) &\Omega^1(X;\hh) \ar[l]_{\d^*}  &\Omega^{n-1}(X; \hh_\perp) \ar[l]_{\d^*\star\d^*} & \Omega^n(X, \hh) \ar[l]_{\d^*} \\
\Omega^0(X; \hh_\perp)  &\Omega^1(X;\hh_\perp) \ar[l]_{\xi^{-1/2}\d^*}  &\Omega^{n-1}(X; \hh_\perp) \ar[l]_{\xi m^2 \star} \ar[dl]^(.6){\d^*_{2}} &\Omega^n(X; \hh_\perp) \ar[l]_{\xi^{-1/2}\d^*} \ar[dl]^{\xi^{1/2}\phi_0} \\
&\Omega^0(X; R_0) \ar[ul]^{\xi^{1/2} \phi_0} &\Omega^{n}(X; R_0) \ar[l]^{\xi^{-1}\d^*\star\d^*} \ar[ul]_(.6){\d^*_{1}} \\
&\Omega^0(X; R_\perp)  &\Omega^{n}(X; R_\perp). \ar[l]_{\star}
}
\]
where $\d^*_{1} = \phi_0\star \d^*$ and $\d^*_{2} = \phi_0\d^*\star$, independent of $\xi$.  
This definition can be extended to a family over $\bb{RP}^1$ as the operator $D^*_{(s:t)}$ defined as 
\[
\xymatrix{
\Omega^0(X; \hh) &\Omega^1(X;\hh) \ar[l]_{\d^*}  &\Omega^{n-1}(X; \hh_\perp) \ar[l]_{\d^*\star\d^*} & \Omega^n(X, \hh) \ar[l]_{\d^*} \\
\Omega^0(X; \hh_\perp)  &\Omega^1(X;\hh_\perp) \ar[l]_{s\d^*}  &\Omega^{n-1}(X; \hh_\perp) \ar[l]_{t^2m^2 \star} \ar[dl]^(.64){st\d^*_2} &\Omega^n(X; \hh_\perp) \ar[l]_{s\d^*} \ar[dl]^{t\phi_0} \\
&\Omega^0(X; R_0) \ar[ul]^{t \phi_0} &\Omega^{n}(X; R_0) \ar[l]^{s^2\d^*\star\d^*} \ar[ul]_(.64){st\d^*_1} \\
&\Omega^0(X; R_\perp)  &\Omega^{n}(X; R_\perp). \ar[l]_{\star}
}
\]
\end{lemma}

The proof is by direct verification.  
Notice that we have seen the operator $D^*_{(0:1)}$ before: it is exactly the homotopy $h$ appearing in Proposition \ref{def_retract_prop} that implements the deformation retraction!

\begin{remark}
This gauge-fixing does not satisfy the requirement that $[D, D^*_\xi]$ is a generalized Laplacian, as requested by \cite{Costello} to apply his results directly.
Instead, the commutator includes a differential operator of order~4.  
Working in the first order formalism for Yang--Mills theory resolves this issue tidily.
\end{remark}

\appendix

\section{The shifted \texorpdfstring{$L_\infty$}{L infinity} structures after gauge symmetry breaking} \label{bracket_appendix}

Let $\rho$ denote the action of $\gg$ on $R$.
As in Section~\ref{general_gauge_breaking_section}, we use the invariant pairing on $R$ to define an equivariant map $\langle -,-\rangle_\gg \colon R \otimes R \to \gg$.

As an $L_\infty$ algebra, the classical BV complex of Yang--Mills--Higgs theory at the point $(0, \phi_0)$ includes non-trivial quadratic and cubic Lie brackets.  They are given by the action of the ghost $\Omega^0(X)$ on $R$ in the second row, together with additional cubic and quartic terms that can be read off from the action functional.  These terms were also described in the abelian case in \cite[Section 4.3]{Zeitlin}.  

First, there are quadratic Lie brackets coming from the cubic term in the action functional.  These are given by the symmetrizations of the following maps:
\begin{itemize}
\item between connection 1-forms, we have
\[
\begin{array}{ccc}
\Omega^1(X, \gg) \otimes \Omega^1(X, \gg)  &\to &\Omega^{n-1}(x, \gg) \oplus \Omega^n(X; R) \\
A_1 \otimes A_2 &\mapsto &\left( [A_1 \wedge \star \d A_2] + \frac 12 \d \star [A_1 \wedge A_2], \rho(A_1) \wedge \star \rho(A_2)\phi_0 \right)
\end{array}
\]
\item between connection 1-forms and $R$-valued scalar fields, we have
\[
\begin{array}{ccc}
\Omega^1(X, \gg) \otimes \Omega^0(X, R)  &\to& \Omega^{n-1}(x, \gg) \oplus \Omega^n(X; R) \\
A \otimes \phi &\mapsto&  (\langle \phi, \star \rho(A)\phi_0 \rangle,  \d \star \rho(A) \phi )
\end{array}
\]
\item between $R$-valued scalar fields, we have
\[
\begin{array}{ccc}
\Omega^0(X, R) \otimes \Omega^0(X, R)  &\to& \Omega^{n-1}(x, \gg) \oplus \Omega^n(X; R) \\
\phi_1 \otimes \phi_2 &\mapsto& (\langle \phi_1, \star \d \phi_2\rangle_\gg, \rho(\langle \phi_1, \phi_2 \rangle) \phi_0 )
\end{array}
\]
\item between connection 1-forms and their antifields, we have
\[
\begin{array}{ccc}
\Omega^1(X, \gg) \otimes \Omega^{n-1}(X, \gg)  &\to &\Omega^{n-2}(x, \gg) \\
A \otimes A^* &\mapsto &[A \wedge A^*].
\end{array}
\]
\end{itemize}
Note that in the penultimate bracket we take the total symmetrization in $\phi_0, \phi_1$ and $\phi_2$.

There are also cubic $L_\infty$ brackets coming from the quartic term in the action functional.  
These are independent of the choice of solution $\phi_0$ and are given by the total symmetrizations of the following maps:
\[
\begin{array}{ccc}
\Omega^1(X, \gg) \otimes \Omega^1(X, \gg) \otimes \Omega^1(X; \gg)  &\to& \Omega^{n-1}(x, \gg)  \\
A_1 \otimes A_2 \otimes A_3 &\mapsto& [A_1 \wedge \star [A_2 \wedge A_3]]
\end{array}
\]
\[
\begin{array}{ccc}
\Omega^1(X, \gg) \otimes \Omega^1(X, \gg) \otimes \Omega^0(X; R)  &\to&  \Omega^n(X; R) \\
A_1 \otimes A_2 \otimes \phi &\mapsto& \rho(A_1) \wedge \star \rho(A_2)\phi \\
\end{array}
\]
\[
\begin{array}{ccc}
\Omega^1(X, \gg) \otimes \Omega^0(X, R) \otimes \Omega^0(X; R)  &\to& \Omega^{n-1}(x, \gg) \\
A \otimes \phi_1 \otimes \phi_2 &\mapsto& (\langle \phi_1, \star \rho(A)\phi_2 \rangle_\gg \\
\end{array}
\]
\[
\begin{array}{ccc}
\Omega^0(X, R) \otimes \Omega^0(X, R) \otimes \Omega^0(X; R)  &\to& \Omega^n(X; R) \\
\phi_1 \otimes \phi_2 \otimes \phi_3 &\mapsto& \rho(\langle \phi_1, \phi_2 \rangle) \phi_3.
\end{array}
\]
There are no higher order terms in the action functional and hence no higher brackets.

\subsection*{Acknowledgements}
We benefited from discussions with a number of people, notably Ingmar Saberi, Brian Williams, and Philsang Yoo, as we chewed on the Higgs mechanism from a BV perspective.
We are also grateful to Eugene Rabinovich for his helpful comments on an earlier draft.
We thank in particular David Carchedi, who gave us extensive feedback on the draft and has discussed further refinements of this approach, and John Huerta, whose careful reading and thoughtful suggestions caught many issues and improved the paper substantially. Subsequent to publication and as part of a careful MathSciNet review, Francis Bischoff gave us feedback that really improved the discussion of the Higgs mechanism, where multiple revisions had led to mismatches between various parts of the text and hence to a confusing exposition.
We thank him warmly for his help!
The National Science Foundation supported O.G. through DMS Grant No. 1812049.
Any opinions, findings, and conclusions or recommendations expressed in this material are those of the authors and do not necessarily reflect the views of the National Science Foundation.

\pagestyle{bib}
\printbibliography

\textsc{University of Massachusetts, Amherst}\\
\textsc{Department of Mathematics and Statistics, 710 N Pleasant St, Amherst, MA 01003}\\
\texttt{celliott@math.umass.edu}\\
\texttt{gwilliam@math.umass.edu}

\end{document}